\documentclass{IEEEtran}
\usepackage{url}
\usepackage{amsmath}
 
\usepackage{amsthm}
\usepackage[linesnumbered,ruled]{algorithm2e}
\usepackage{verbatim}
\usepackage{graphicx}
\usepackage{subfigure}
\usepackage{cite}
\newtheorem{theorem}{Theorem}
\newtheorem{lemma}{Lemma}

\newtheorem{prob}{Problem}

\theoremstyle{definition}

\newtheorem{deff}{Definition}

\newtheorem{prop}{Property}
\newtheorem{propp}{Proposition}

\graphicspath{{figures/}}
\usepackage{amsmath}
\usepackage{mathrsfs}
\usepackage{amssymb}
\usepackage{chngcntr}
\newcommand{\shrinkfig}{\vspace{-0.3cm}}

\title{Scalable Bicriteria Algorithms for the Threshold Activation Problem in Online Social Networks
\thanks{\footnotesize \textcopyright 2017 IEEE.
Personal use of this material is permitted. Permission 
from IEEE must be 
obtained for all other uses, in any current or 
future media, including 
reprinting
/republishing this material for 
advertising or promotional purposes,
creating new 
collective works,
or resale or 
redistribution to servers or lists, or reuse of any 
copyrighted 
component of this work in other works.\hfill }}
 \author{Alan Kuhnle, Tianyi Pan, Md Abdul Alim, and My T. Thai\\
  Department of Computer \& Information Science \& Engineering\\
  University of Florida\\
  Gainesville, Florida 32611\\
  Email: \{kuhnle, tianyi, alim, thai\}@cise.ufl.edu }

\begin{document}

\addtolength{\topmargin}{0.04in} \maketitle 
\begin{abstract} 
We consider the Threshold Activation Problem (TAP): 
given social network $G$ and positive 
threshold $T$, find a minimum-size seed set $A$ that can 
trigger expected activation of at least $T$. 
We introduce the first scalable, parallelizable algorithm
with performance guarantee 
for TAP suitable for datasets with millions of nodes
and edges; we exploit the bicriteria nature of 
solutions to TAP to allow the user to control the running
time versus accuracy of our algorithm through a parameter
$\alpha \in (0,1)$: 
given $\eta > 0$, with probability $1 - \eta$ our algorithm
returns a solution $A$ with expected activation greater than
$T - 2 \alpha T$, and the size of the solution $A$ is within 
factor $1 + 4 \alpha T + \log ( T )$
of the optimal size. The algorithm runs in time
$O \left( \alpha^{-2}\log \left( n / \eta \right) (n + m) |A| \right)$,
where $n$, $m$, refer to the number of nodes, edges
in the network. The performance guarantee holds
for the general triggering model of internal influence
and also incorporates external influence, provided 
a certain condition is met on the cost-effectivity of seed
selection. 
\end{abstract}
\section{Introduction} \label{sect:intro}
With the growth of online social networks,
viral marketing where influence
spreads through a social network 
has become a central research topic. Users of a social
network can \emph{activate} their friends by
influencing them to adopt certain behaviors or products.
In this context, the \emph{influence maximization} (IM) problem
has been studied extensively \cite{Kempe2003,dinh2014cost,Cohen2014a,Borgs2014,zhang2013maximizing,zhang2016least,dinh2012cheap,Zhang2016}:
given a budget, the IM problem is to find a \emph{seed set}, or
set of initially activated users, within the budget
that maximizes the expected influence.
Much recent work \cite{Borgs2014,Tang2015,Tang2014,Cohen2014a}
has developed scalable algorithms
for IM that are capable
of running on social networks with millions of nodes while retaining
the provable guarantees on the quality of solution; namely, that the 
algorithm for
IM will produce a solution with expected influence within 
$1 - 1/e$ of the optimal activation.

However, a company with a specific target in mind may adopt a more flexible
approach to its budget: instead of having a fixed budget $k$ for the seed 
set, it is natural to minimize
the size of the seed set while activating
a desired threshold $T$ of users within the network. For example, suppose 
a company desires a certain level of exposure on social media; such
exposure could boost the sales of any of its products. Alternatively, suppose
a profit goal $T$ for a product must be met with the least expense possible.
Thus, we consider the following \emph{threshold activation} problem (TAP):
given a threshold $T$, minimize the size of the set of seed users in order to activate at least
$T$ users of the network in expectation. 

Goyal et al. \cite{Goyal2013} provided bicriteria performance guarantees for
a greedy approach to TAP based upon Monte Carlo sampling at each 
iteration to select the best seed node, an algorithm 
reminiscent of the greedy algorithm for IM in
Kempe et al. \cite{Kempe2003}; this approach is inefficient
and impractical for large networks.
Although TAP is related to IM, scalable solutions that already exist
for IM are unsuitable for TAP: the TIM \cite{Tang2014} and
IMM \cite{Tang2015} algorithms require knowledge of 
the size $k$ of the seed set ahead of time; 
the SKIM algorithm \cite{Cohen2014a}
for average reachability has been shown to be effective for IM
in specialized settings, 
but it is unclear how to apply
SKIM to more general situations or to TAP while retaining performance guarantees.

Moreover, empirical studies have shown that in the viral
marketing context, it is insufficient to consider internal
propagation alone; 
external activation, i.e.
activations that cannot be explained by the network
structure, play a large role 
in influence propagation events 
\cite{Goel2015,Gonzalez-Bailon2013,Myers2014,Myers2012}, 
and recent works on scalable algorithms for IM 
\cite{Tang2014, Tang2015, Cohen2014a} have neglected the consequences
of external influence.
For internal diffusion, two basic models
have been widely adopted, the independent cascade and linear 
threshold models; Kempe et al. \cite{Kempe2003} showed these
two models are special cases of the \emph{triggering model},
a powerful, general model that
has desirable properties in a
viral marketing context. 

Motivated by the above observations, the main contributions of this work are:
\begin{itemize}
  \item
    We establish a new connection between the 
    triggering model
    and a concept of generalized reachability that
    allows a natural combination 
    of external influence
    with the triggering model. We show any instance of 
    the triggering
    model combined with any model of 
    external influence is monotone
    and submodular.
  \item
    We show how to use the generalized reachability 
    framework to efficiently estimate 
    the expected influence of the triggering model
    combined with external influence, leveraging
    scalable estimators of average reachability
    by Cohen et al.
    \cite{Cohen1997,Cohen2014a,Cohen2009}. 
    This efficient estimation results
    in a parallelizable algorithm (STAB)
    for TAP 
    with performance guarantee 
    in terms of user-adjustable trade-offs between efficiency
    and accuracy. The desired accuracy is input
    as parameter $\alpha \in (0,1)$ which determines
    running time as $O \left( \alpha^{-2}\log (n / \eta ) (n + m) |A| \right)$,
    where $n, m$ are number of nodes, edges in the network, and
    $A$ is the seed set returned by STAB. With probability
    $1-\eta$, the expected activation is guaranteed to be 
    with $2 \alpha T$ of threshold $T$,
    and the
    size of the seed set $A$ is guaranteed to
    be within factor    
    $1 + 4 \alpha T + \log ( T )$ of the
    optimal size. If the cost-effectivity of seed
    selection
    falls below 1, this performance bound may not hold;
    we provide a looser bound for this case.

  \item 
    Through a comprehensive set of experiments, 
    we demonstrate that on large networks,
    STAB not only returns a better solution to TAP,
    but it runs faster than existing algorithms for 
    TAP and algorithms for IM adapted
    to solve TAP, often by factors of more than $10^3$ even for 
    the state-of-the-art IMM algorithm \cite{Tang2015}.
    In addition, we investigate the effect of
    varying levels of external influence on the solution 
    of STAB.

\end{itemize}
The rest of this paper is organized as follows: in Section
\ref{label:model}, we introduce models of influence, including
the triggering model and our concept of generalized reachability.
We prove these two concepts are equivalent.  In Section
\ref{sect:TA}, we formally define TAP and first prove
bicriteria guarantees in a general setting. 
Next, we employ the generalized reachability
concept to show how combination of triggering model and
external influence can be estimated efficiently, and we
present and analyze STAB, our
scalable bicriteria algorithm.
In Section \ref{sect:exp}, we analyze STAB
experimentally and compare with prior work. 
We discuss related work in Section 
\ref{sect:rw}. \vspace{-0.2cm}
\section{Models of influence}\label{label:model}
A social network can be modeled as a directed
graph $G =(V,E)$, where $V$ is the set of users
and directed edges $(u, v) \in E$ denote social connections,
such as friendships, between the users 
$u, v$. In this work, we study the propagation of influence
through a social network; for example,
say a user on the Twitter network posts a message to
her account; this message may be reposted by the friends
of this user, and as such it propagates across the social network.
In order to study such events from a theoretical standpoint,
we require the concept of a \emph{model of influence propagation}.

Intuitively, the idea of a model 
of influence propagation in a network is
a way by which nodes can be activated 
given a set of seed nodes. In this work, we 
use $\sigma$ to denote a model of influence propagation.
Such a model is usually probabilistic, and
the notation $\sigma(S)$ will denote the
expected number of activations under the model $\sigma$
given seed set $S \subseteq V$.  Kempe et al.  \cite{Kempe2003}
studied
a variety of models in their seminal work on influence
propagation on a graph, including
the Independent Cascade (IC) and Linear Threshold (LT)
models.
For completeness, we briefly describe these two models. 
An instance 
of influence propagation on a graph $G$ follows the IC model if a weight can be assigned
to each edge such that the propagation probabilities can be computed as
follows: once a node $u$ first becomes active, 
it is given a single chance to activate each currently inactive neighbor $v$ with probability proportional to 
the weight of the edge $(u,v)$. 
In the LT model each network user $u$ has an associated threshold $\theta(u)$ chosen uniformly from $[0,1]$ which determines how much influence (the sum of the weights of incoming edges) is required to activate $u$. $u$ becomes active if the total influence from its activated neighbors exceeds the threshold $\theta(u)$.

These well-studied models are both examples of the \emph{Triggering Model},
also introduced in \cite{Kempe2003}:
    Each node $v \in G$ independently chooses a random ``triggering
    set'' $T_v$ according to some probability distribution over subsets of its neighbors.
    A seed set $S$ is activated at time $t = 0$; a node $v$ becomes active at
    time $t$ if any node in $T_v$ is active at time $t - 1$.

Two important properties that it is desirable for 
a model $\sigma$ of influence propagation to satisfy
are firstly the \emph{submodularity property} of the
expected activation function: for any sets $S, T \subseteq V$, $\sigma (S \cup T) + \sigma (S \cap T) \le \sigma (S) + \sigma (T)$, and secondly the \emph{monotonicity property}: 
if $S \subseteq T \subseteq V$,
$\sigma (S) \le \sigma (T).$
These properties together allow 
a greedy approach to have a performance
ratio for the influence maximization problem (IM): given $k$, 
find a seed set $A$ of size $k$ such that the expected activation
of $A$ is maximized. Kempe et al. \cite{Kempe2003} 
showed that the triggering
model is both submodular and monotonic. 
Both properties are also important in proving performance
guarantees for TAP, the problem studied in this work
and defined in Section \ref{sect:TA}.

It is $\# \mathbf{P}$-hard to compute the exact influence
of a single seed node under even the restricted version of IC
where each edge is assigned the probability $0.5$ \cite{Chen2010a}.
Therefore, it is necessary to estimate the value of
$\sigma (S)$ by sampling the probability distribution
determined by $\sigma$. Sampling efficiently such that
the estimated value $\hat{ \sigma }(S)$ 
satisfies $| \hat{ \sigma }(S) - \sigma (S) | < \varepsilon$ for all
seed sets is a difficult
problem; we discuss 
this problem when $\sigma$ is
a combination of the
triggering model and external influence 
in Section \ref{sect:est_sigma}.
Because of the errors associated with estimating 
the value $\sigma (S)$,
we introduce slightly
generalized versions of the above two properties.  First,
let us define
$\Delta_x \sigma(A) = \sigma \left( A \cup \{ x \} \right) - \sigma (A)$,
for any model $\sigma$, and subset $A \subset V$.
The following property
is equivalent to $\sigma$ satisfying submodularity and monotonicity
together.
\begin{prop}[Submodularity and monotonicity] \label{prop:sm_mono}
For all $A \subseteq B$, and $x \in V$,
$\Delta_x \sigma (A) \ge \Delta_x \sigma (B).$
\end{prop}
Let us define a $\sigma$ to be $\varepsilon$-approximately
monotonic and submodular if
the following property is satisfied
instead:
\begin{prop}[$\varepsilon$-submodularity and monotonicity] \label{prop:approx_sm_mono}
Let $\varepsilon > 0$. For all $A \subseteq B$, and $x \in V$,
$\varepsilon + \Delta_x \sigma(A) \ge \Delta_x \sigma(B).$
\end{prop} 
\subsection{Triggering model from the perspective of generalized reachability}
Next we define a class of influence propagation models
that naturally generalize the notion of reachability in
a directed graph: that is, these models generalize 
the simple model that a node is activated by
a seed set $S$ if it
is reachable from the seed set by edges in the graph.
Somewhat surprisingly, the triggering model is equivalent
to this notion of generalized reachability.

Suppose instead of a single directed graph
$G$, we have a set of directed graphs
$\{(G_i = (V, E_i), p_i)\}_{i=1}^l$ 
on the same vertex set $V$, and associated probabilities
such that $\sum_{i=1}^l p_i = 1$. 
Then define an influence propagation model
$\sigma$ in the following way: when a seed set $S$
is activated, graph $G_i$ is chosen with probability
$p_i$. Then, influence is allowed to propagate from seed
set $S$ in $G_i$ according to the directed edges of
$G_i$. Let $\sigma_i (S)$
be defined as the number of vertices reachable from $S$
in $G_i$. Then the expected activation of
a seed set $S$ is given by
$\sigma (S)  = \sum_{i=1}^l p_i \sigma_i (S).$
We will term a model $\sigma$ of this form 
a model of \emph{generalized reachability},
since it generalizes the notion of simple
reachability on a directed graph. We have
the following important proposition:
\begin{propp}
  Generalized reachability is equivalent
  to the Triggering Model.
  \label{reach_trig}
\end{propp}
\begin{proof} 
  Suppose we have an instance $\sigma$ of the
  triggering model. For each node $v$, triggering
  set $T_v$ is chosen independently with 
  probability $p(T_v)$. If all nodes choose a triggering
  set, then define graph $H_T$ for this choice by
  adding directed edges $(u, v)$ for each $u \in T_v$.
  Assign graph $H_T$ probability $\prod_v p( T_v )$
  and we have instance of generalized reachability
  with the same expected activation.

  Conversely, suppose $\sigma$ is an instance of
  generalized reachability. For $v \in V$, 
  let $T_v$ be any subset
  of nodes excluding $v$ itself. Assign $p(T_v)$
  to be the sum of the probabilities of the graphs
  in which the in-neighbors of $v$ are exactly
  the set $T_v$. Then we have an instance of the
  triggering model. 
\end{proof}

\subsection{External influence}
In this section, we outline our model of
external influence $\sigma_{ext}$ in a social network.
We wish to capture the idea that users in the network
may be activated by a source external to the network;
that is, these activations do not occur through friendships
or connections within the social network.
The most general model of external influence 
is simply an arbitrary probability distribution
on the set of subsets of nodes. That is, for any $S \subset V$,
there is a probability $p_S$ that $S$ is activated
from an external source. In this work, we adopt this
model and denote such a model
of external influence by $\sigma_{ext}$. In order to 
consider both external and internal influence in
our social networks, we next define the concept of combining
models of influence together.

\begin{deff}[Combination of $\sigma_1, \sigma_2$]
  Let $\sigma_1$ and $\sigma_2$ be two models of influence
  propagation. We define the combination $\sigma$ of these
  two models in the following way: At any timestep $t$, if 
  set $A_t \subset V$ is activated
  new nodes may be activated from $A_t$ according to either
  $\sigma_1$ or $\sigma_2$; that is, if in the next timestep
  $\sigma_1$ activates $A^1_{t + 1}$ and $\sigma_2$ activates
  $A^2_{t + 1}$, then $A_{t+1} \equiv A^1_{t+1} \cup A^2_{t+1}$.
  We denote the combination model
  as $\sigma_1 \oplus \sigma_2$ and
  write $\sigma_1 \oplus \sigma_2 (A)$ 
  for the expected number of activations
  resulting from seeding $A$ under this model.
\end{deff}
In this work, we are most interested in combining 
external activation with the triggering model. That is,
if $\sigma_{ext}$ is a model of external influence as
defined above, and $\sigma_{trig}$ is an instance of the
triggering model, then we consider 
$\sigma_{trig} \oplus \sigma_{ext}$.
\subsubsection{Submodularity of $\sigma_{trig} \oplus \sigma_{ext}$}
In this subsection, we establish the submodularity and monotonicity
of $\Phi \equiv \sigma_{trig} \oplus \sigma_{ext}$. First, we require the following proposition.
\begin{propp} \label{prop:rel_act}
  Let $\sigma$ be a submodular, monotone-increasing
  model of influence propagation.
  Define $\sigma_A (T)$ to be the expected influence
  of seeding set $T$ when $A$ is already
  activated; that is $\sigma_A (T) = \sigma (A \cup T)$.
  Then, for any $A$, $\sigma_A$ is submodular and monotone increasing.  
\end{propp}
\begin{proof}
  Let $T_1, T_2 \subseteq V$. We prove submodularity only,
  monotonicity is similar.
  \begin{align*}
    &\sigma_A (T_1 \cup T_2) + \sigma_A (T_1 \cap T_2) \\
    &= \sigma (T_1 \cup T_2 \cup A) + \sigma (T_1 \cap T_2 \cup A ) \\
    &= \sigma ( (T_1 \cup A) \cup (T_2 \cup A) ) + \sigma ( (T_1 \cup A) \cap (T_2 \cup A) ) \\
    &\le \sigma (A \cup T_1) + \sigma (A \cup T_2) \\
    &= \sigma_A (T_1) + \sigma_A (T_2). \qedhere
  \end{align*}
\end{proof}
\begin{theorem}
  Let 
  $\sigma_{trig}$ be any instance of the Triggering Model, and
  $\sigma_{ext}$ be any instance of external influence. Then,
  $\Phi \equiv \sigma_{trig} \oplus \sigma_{ext}$ is submodular
  and monotonic.
\end{theorem}
\begin{proof}
  For any $A \subseteq V$, let $p_A$ be the probability that
  $A$ is activated via $\sigma_{ext}$, the external influence.
  Then, 
  $\Phi(S) = \sum_{A \subset V} p_A \sigma_{trig} ( A \cup S ).$
  By Prop. \ref{prop:rel_act}, $\sigma_{trig, A} (S) = \sigma_{trig} (A \cup S)$
  for fixed $A$ is monotone and submodular. Since a non-negative 
  combination of submodular and monotonic 
  set functions is also submodular and
  monotonic, the result follows. 
\end{proof}

\section{Threshold Activation Problem} \label{sect:TA}
The framework has been established to define
the problem we consider in this work. 
We suppose a company wants to minimize the
number of seed users while expecting a certain 
level of activation 
in the social network.
Formally, we have
\begin{prob}[Threshold activation problem (TAP)] \label{prob:TA}
  Let $G = (V,E)$ be a social network, $n = |V|$.
  Given influence propagation model $\sigma$ 
  and threshold $T$ such that
  $0 \le T \le n$,
  find a subset $A \subseteq V$ such that
  $|A|$ is minimized with $\sigma(A) \ge T$.
\end{prob}
First, we consider performance guarantees for
a greedy approach to TAP with only
the assumption that the influence 
propagation $\sigma$ is $\varepsilon$-submodular.
We do not discuss how to sample $\sigma$ for
these results.
Subsequently, we specialize to the case when
$\sigma$ is the the estimated value of the 
combination of the triggering model
and external influence, which is approximately submodular
up to the error of estimation, and we show how to
efficiently estimate to a desired accuracy
in Sections \ref{sect:est_sigma} and \ref{sect:est_tau}. We
detail the algorithm STAB in Section \ref{sect:alg_description}, 
a scalable algorithm with performance
guarantees for TAP utilizing this estimation and
analysis. \vspace{-0.3cm}

\subsection{Results when $\sigma$ is $\varepsilon$-submodular}
We analyze the greedy algorithm to solve TAP,
which adds a node that maximizes the marginal 
gain at each iteration to its chosen seed set
$A$ until $\sigma (A) \ge T$, at which point
it terminates. One might imagine in analogy
to the set cover problem that there would
be a $\log n$-approximation ratio for the 
greedy algorithm for TAP -- however, this result
only holds for \emph{integral} submodular
functions \cite{Goyal2013}. We next give a
bicriteria result for the greedy algorithm
that incorporates the error $\varepsilon$
inherent in $\varepsilon$-approximate submodularity
into its bounds.
In this context,
a bicriteria guarantee means that
the algorithm is allowed to violate the
constraints of the problem by a specified amount,
and also to approximate the solution to the problem.
In the viral marketing context, this means that
we may not activate the intended threshold $T$
of users, but we will guarantee to activate
a number close to $T$. Furthermore, we will not
achieve a solution of minimum size, but there
is a guarantee on how large the solution returned
could be.

\begin{theorem}
  Consider the TAP problem for $\sigma$ when
  $\sigma$ is $\varepsilon$-approximately
  submodular.
  Then the greedy algorithm that terminates when
  the marginal gain is less than $1$ 
  returns a solution $A$
  of size within factor 
  $1 + \varepsilon + \log \left( \frac{T}{OPT} \right)$ of the optimal solution
  size, $OPT$, and the solution $A$ satisfies
  $\sigma (A) \ge T - (1 + \varepsilon) \cdot OPT.$   \label{thm:apx_sm} 
\end{theorem}
\begin{proof} 
Let $A_i = \{a_1, \ldots, a_i \}$ be the
greedy solution after $i$ iterations, and let $A_g$ be the 
final solution
returned by the greedy algorithm. Let $o = OPT$ be the size of an optimal
solution $C = \{c_1, \ldots, c_{o} \}$ satisfying $\sigma ( C ) \ge T$.
Then
\begin{align}
  T - \sigma( A_i ) &\le \sigma(A_i \cup C) - \sigma(A_i) \nonumber \\
  &= \sum_{j = 1}^{o} \Delta_{c_j} \sigma \left( A_i \cup \{ c_1, \ldots, c_{j-1} \} \right) \nonumber \\
  &\le \sum_{j = 1}^{o} \Delta_{c_j} \sigma ( A_i ) + o\varepsilon \qquad \text{(by Property \ref{prop:approx_sm_mono})} \nonumber \\
  &\le o \cdot \left[ \sigma( A_{i + 1} ) - \sigma( A_i ) + \varepsilon \right]. \label{increase_bound}
\end{align}
Therefore, $T - \sigma( A_{i + 1} ) - \varepsilon \le \left( 1 - \frac{1}{o} \right) (T - \sigma( A_i ) ).$
Then
\begin{align}
  T - \sigma(A_i) &\le T \left( 1 - \frac{1}{o} \right)^i + \varepsilon \sum_{j=0}^{i - 1} \left( 1 - \frac{1}{o} \right)^j \nonumber \\
  &\le T \left( 1 - \frac{1}{o} \right)^i + \varepsilon o. \label{ineq_diff}
\end{align}
From here, there exists an $i$ such that the
following differences satisfy
\begin{align}
T - \sigma(A_{i}) &\ge o( 1 + \varepsilon) \text{, and} \label{ineq_diff2} \\
T - \sigma(A_{i + 1}) &< o( 1 + \varepsilon ), 
\label{ineq_size_bd}
\end{align}
Thus, by inequalities (\ref{ineq_diff}) and
(\ref{ineq_diff2}),
$o \le T \exp \left( \frac{-i}{o} \right),$
and 
$ i \le o \log \left( \frac{ T}{ o } \right).$
By inequality
(\ref{ineq_size_bd}) and the assumption on the
termination of the algorithm, the greedy
algorithm adds at most
$o(1 + \varepsilon )$ more elements, so
$g \le i + o(1 + \varepsilon ) \le o \left( 1 + \varepsilon  + \log \left( \frac{T}{o} \right) \right).$
Finally, if the algorithm terminates before $\sigma ( A_g ) \ge T$,
then the marginal gain is less than $1$. Hence, by
(\ref{increase_bound}), 
$\sigma ( A ) \ge T - (1 + \varepsilon) o.$ \end{proof}
Notice that the above argument requires only
  that $\sigma$ is a $\varepsilon$-submodular
  set function; in particular, it
  did not use the fact that $\sigma$ represents expected influence
  propagation on a social network.

\subsubsection{Approximation ratios} \label{sect:apx}
Next, we consider ways in which the
bicriteria guarantees of Theorem
\ref{thm:apx_sm} can be improved. 
In viral marketing, we may suppose
a company seeks to choose a 
threshold $T$ such that the marginal gain to reach
$T$ is always at least $1$; seeding nodes 
with a marginal gain of less than $1$ would
be cost-ineffective. In other words,
it would cost more to seed a node than
the benefit obtained from seeding it.
There is little point in activating 
$T$ users if the marginal gain drops
too low; intuitively, the company
has already activated as many as
it cost-effectively can.

We term this assumption
the \emph{cost-effectivity assumption} (CEA):
  In an instance of TAP, 
  if $B \subseteq V$ such that
  $\sigma (B) < T$,
  there always exists a node $u$ such that
  $\Delta_u \sigma (B) \ge 1$.
Under CEA,
the greedy algorithm in Theorem \ref{thm:apx_sm}
would be an approximation algorithm; that is,
it would ensure $\sigma (A) \ge T$, with the
same bound on solution size as stated in the theorem. 
To see this fact, notice that once
inequality (\ref{ineq_size_bd}) above
is satisfied, the algorithm must add
at most $o(1 + \varepsilon)$ more elements
before $\sigma (A) \ge T$, by CEA.

More generally, each node $v \in V$
has an associated \emph{reticence} $r_v \in [0,1]$;
$r_v$ is the probability that $v$ will remain
inactive even if all of the neighbors of $v$
are activated. Then we have the following
theorem, whose proof is analogous to Theorem
\ref{thm:apx_sm}.
\begin{theorem}
  Let $\sigma$ be $\varepsilon$-approximately
  submodular and let $r^* = \min_{v \in V} r_v$. Suppose
  $r^* > 0$. Then the greedy algorithm for
  TAP is an approximation algorithm which
  returns solution $A$ within factor
  $\frac{ 1 + \varepsilon }{ r^* } +\log \left( \frac{T}{OPT} \right)$  
  of optimal size.
\end{theorem}

\subsection{Scalable bicriteria algorithm for $\sigma_{int} \oplus \sigma_{ext}$} \label{sect:STAB}
In this subsection, we detail the scalable 
bicriteria Algorithm \ref{alg:bicriteria}
for TAP when
the propagation is given by 
an instance of the triggering model
in the presence of external influence; that is, 
when $\sigma = \sigma_{trig} \oplus \sigma_{ext}$. 
We describe our scalable algorithm STAB first and then 
discuss the necessary sampling and estimation techniques
in the subsequent sections.
\subsubsection{Description of algorithm} \label{sect:alg_description}
As input, the algorithm takes
a graph $G = (V,E)$ representing a social network,
an external influence model $\sigma_{ext}$, internal
influence model $\sigma_{trig}$, an instance of the
triggering model, and the desired threshold
of activated users $T$. In addition, the user specifies
the \emph{fractional error} $\alpha$, on which the running time
of the algorithm and the accuracy of the solution depend.
Using $\alpha$, in line 1 the algorithm first determines $\ell$, the
number of graph samples it requires according to 
Section \ref{sect:est_sigma}.

Next in the for loop on line 2, the algorithm 
constructs a collection of \emph{oracles} which will be used
to estimate the average reachability in the sampled 
graphs $H_i$, which is used to approximate the expected
influence. Each graph $H_i$ is needed only while updating the
oracle collection in iteration $i$; 
once this step is completed, $H_i$ may be safely discarded.
Since the samples are independent, this process is completely
parallelizable.

Once the set of oracles has been constructed, a greedy algorithm
is performed in attempt to satisfy the threshold
$T$ of expected activation with a minimum-sized seed set.
The estimation in line 11 may be done in one of two ways:
using estimator $C1$ or $C2$; 
both are described in detail in Section \ref{sect:est_tau}. The
estimator chosen has a strong effect on both the
running time and performance of the algorithm: given 
the same oracles, $C1$ can be computed in time
$O( k )$, and in practice is much faster to compute 
than $C2$. However, the quality of $C1$ degrades with
the size of seed set. On the other hand, $C2$ takes
time $O( k |B| \log |B| )$ time to compute, where $B$
is the seed set for which the average reachability
is estimated; our experimental results show that 
the quality of $C2$ is vastly superior to $C1$ for
larger seed set sizes $|B|$; however its running
time increases. 

This algorithm achieves the following guarantees on performance:
\begin{theorem} \label{thm:guarantee}
  Suppose we have an instance of TAP with 
  $\sigma = \sigma_{trig} \oplus \sigma_{ext}$
  and that $T$ has been chosen such that CEA
  holds.

  Then, if $\eta > 0$, by choosing $\delta = \eta / n^3$,
  the solution 
  returned by Alg. \ref{alg:bicriteria} satisfies the
  following two conditions, with probability at least
  $1 - \eta$:
  \begin{enumerate}
  \item
    $\sigma (A) \ge T - 2\alpha T$
  \item
    If $A^*$ is an optimal solution satisfying 
    $\sigma (A^*) \ge T $,
    $\frac{ |A| }{ |A^*| } \le 1 + 4 \alpha T + \log ( T ).$
  \end{enumerate}
  If Assumption CEA is violated, the algorithm can
  detect this violation by terminating when the marginal
  gain drops below 1. In this case, bound 1 above
  becomes $\sigma( A ) \ge T - (1 + \varepsilon) |A^*|$.

  Using estimator $C1$ for average reachability
  in line 11 yields
  running time $O \left( \log ( n / \eta ) (n + m) |A| / \alpha^2 \right)$.
  If estimator $C2$ is used, a factor of $|A| \log |A|$
  is multiplied by this bound.
\end{theorem}
\begin{proof}
  Let $A^*$ be a seed set of minimum size satisfying
  $\sigma (A^*) \ge T$. Then, as discussed in Section
  \ref{sect:est_sigma}, if $\delta = \eta / n^3$,
  then with probability at least $1 - \eta$, 
  $A^*$ satisfies
  $\hat{ \sigma } (A^*) \ge T - \alpha T$ by
  the choice of $\ell$ in Alg. \ref{alg:bicriteria},
  and the analysis in Section \ref{sect:est_sigma}.
  Hence $|A^*| \ge |B^*|$, where
  $B^*$ is a set of minimum size satisfying
  $\hat{\sigma} (B^*) \ge T - \alpha T$. 

  Notice $\hat{\sigma}$ is 
  $4 \alpha T$-approximately submodular on 
  the sets considered by the greedy algorithm with
  probability at least $1 - \eta$.
  By Section \ref{sect:apx}, in this case,
  the solution
  $A$ returned by Alg. \ref{alg:bicriteria}
  satisfies $|A| \le (1 + 4 \alpha T + \log T) |B^*| \le (1 + 4 \alpha T + \log T) |A^*|$.
  Furthermore, 
  $\sigma ( A ) \ge \hat{ \sigma }(A) - \alpha  \ge T - 2 \alpha T,$
  if CEA holds, otherwise the alternate bound follows
  from Theorem \ref{thm:apx_sm}.

  Next, we consider the running time of Alg. \ref{alg:bicriteria}.
  Let $m$ be the number of edges that have a nonzero
  probability to exist in one of the reachability
  instances ($m$ is at most the number of edges in
  input graph $G$).
  Lines 2 and 3 clearly take time $\ell (n + m)$.
  By Cohen et al. \cite{Cohen2014a}, line 4
  takes time $O( k \ell m )$.  The while
  loop on line 9 executes exactly $|A|$
  times, and the for loop on line 10 requires
  time $O(kn)$ if estimator $C1$ is used,
  and time $O(k|A|\log |A|n)$ if estimator $C2$
  by Section \ref{sect:est_tau}.
  By the choices of $k$ and $\ell$ on line 1,
  we have the total running times bounded as stated. \vspace{-0.3cm}
\end{proof}

\begin{algorithm}[h!]
 \KwData{$G, \sigma_{trig}, \sigma_{ext}, T, \alpha, \delta$}
 \KwResult{ Seed set $A$ }
 Choose $\ell \ge \log (2 / \delta) / \alpha^2$,
 $k \ge (2 + c) \log n (\alpha T)^{-2}$ 
 as discussed in Section \ref{sect:est_sigma}\;
 \For{$i = 1$ to $\ell$}{ 
   Sample graph $G_i$ from $\sigma_{trig}$\;
   Sample external seed set $A_i^{ext}$ for $\sigma_{ext}$\;
   Construct $H_i = G_i - \sigma_i (A_i^{ext})$,
   and store the value of $| \sigma_i (A_i^{ext}) |$,
   as described in Section \ref{sect:tau}\;
   Update the oracle collection $\{ X_u : u \in V \}$
   for $\hat{ \tau }$ according to $H_i$ 
   as in Section \ref{sect:tau}\;
 }

 $A \equiv \emptyset$, $O \equiv \sum_{i=1}^\ell 
 \left| \sigma_i (A_i^{ext} ) \right| / \ell$\;
 \While{$\hat{ \sigma } (A) < T - \alpha T $}{
   \For{$u \in V$}{
     Estimate 
     $\Delta_u = \{ \hat{\tau}(A \cup \{u\} ) - \hat{\tau} (A) \}$ 
     using one of the
     estimators $C1$ or $C2$ as described in Sections
     \ref{sect:est_tau}, \ref{sect:est_C2}\;
   }
   $A = A \cup \{ u^* \}$, where $u^*$ maximizes $\Delta_u$\;  
   Compute $\hat{\sigma}(A) = \hat{\tau}(A) + O$
   by Lemma \ref{lem:sig_tau}\;
}
 \caption{Scalable TAP Algorithm with Bicriteria guarantees (STAB)}
 \label{alg:bicriteria} 
\end{algorithm}

\subsubsection{Estimation of $\sigma_{trig} \oplus \sigma_{ext} (A)$}
\label{sect:est_sigma}
Let $\sigma_{trig}$ be a model of
internal influence propagation, which in
this section will be an instance of the triggering model.  Let
$\sigma_{ext}$ be the model of external influence activation,
and let $\sigma = \sigma_{trig} \oplus \sigma_{ext}$ be the
combination of the two as defined above.
We use the
following version of Hoeffding's inequality ($T$ will be the
threshold that is input to Alg. \ref{alg:bicriteria}).
\begin{theorem}[Hoeffding's inequality]
  Let $X_1, \ldots, X_\ell$ be independent random variables
  in $[0,T]$.  Let $\overline{X}$ be the empirical mean
  of these variables,
  $\overline{X} = \frac{1}{ \ell } \sum_{i = 1}^\ell X_i.$
  Then for any $t$
  $Pr \left( \left| \overline{X} - E[ \overline{X} ] \right| \ge t \right) \le 2\exp \left( - \frac{ 2 \ell t^2 }{T^2} \right) .$
\end{theorem}

Since $\sigma_{trig}$ is an instance of the triggering model,
by Theorem \ref{reach_trig}, 
there exists a set $\{ (G_j, p_j) : j \in J \}$
such that for any seed set $A$, 
$\sigma_{trig} ( A ) = \sum_{j \in J} p_j \sigma_{j} (A),$
where $\sigma_j$ is simply the size of the set of nodes
reachable from $A$ in $G_j$. 
Thus, if $A$ is fixed, then 
by taking independent samples of graphs
from the probability distribution on 
$\{ G_j \}$, we get independent samples of 
$\sigma_i (A)$.

In the general model of external influence presented above,
every seed set $B \subset V$ has a probability $p_{B}$ of
being activated by the external influence. By
sampling from this distribution on subsets of nodes, and
independently sampling as above from $\sigma_{trig}$,
 $\sigma (A)$ could be computed exactly in the following way:
$\sigma(A) = \sum_{B \subset V} \sum_{j \in J} p_j p_B \sigma_j (A \cup B )$.
In most cases, this sum cannot be computed in polynomial time,
and certainly it has $\Omega (2^n)$ summands. 
Accordingly, we estimate its value by independently
sampling $\ell$ externally activated 
sets $\{ A_{ext}^j \}$ 
from $\sigma_{ext}$; we also independently
sample $\ell$ reachability graphs $G_j$
according to $\sigma_{trig}$ and estimate by averaging
the size of the reachability from a seed set $A$ in this 
context:
$ \sigma (A) \approx \hat{\sigma} (A) \equiv \frac{1}{\ell} \sum_{j = 1}^\ell \sigma_{j} \left( A \cup A_j^{ext} \right) .$
To estimate $\sigma (A)$ in this way within error 
$\alpha T$ with probability 
at least $1 - \delta$ from Hoeffding's 
inequality we require
$\label{lbound} \ell \ge \frac{ \log (2 / \delta ) }{2 \alpha^2 }$
such samples. 

Now, in our analysis of the greedy algorithm 
in Theorem \ref{thm:apx_sm} only at most $n^3$
sets were considered. All that is required for
the analysis to be correct as it that those
sets were estimated within the error
$\alpha T$; if $\delta < \eta / n^3$,
where $\eta < 1$, then by the union bound,
with probability at least $1 - \eta$,
the analysis for Theorem \ref{thm:apx_sm} holds.
In practice, we were able to get good results
with much higher values of $\delta$, see
Section \ref{sect:exp}.


\subsubsection{Estimation of $\hat{\sigma}$}

In the previous section, we describe an approach to
estimate $\sigma (A)$, based upon independent samples 
$\{ G_i \}$ of graphs from the triggering model,
and independent samples of externally activated nodes
$\{ A_i^{ext} \}$. Next, we need to compute the
value of the estimator $\hat{ \sigma }(A)$. One method
would be to compute it directly using breadth-first
search from the sets $A \cup A_i^{ext}$ in each
graph $G_i$. This method would unfortunately add a
factor of $\Omega (n + m)$ to the running time, which would
result in a running time of Alg. \ref{alg:bicriteria}
of $\Omega (n^2)$, too large for our purposes. Thus, we
would like to take advantage of estimators 
formulated by Cohen et al. \cite{Cohen2014a} for
the average reachability of a seed set 
over a set of graphs. However, because the external
seed sets $A^{ext}_i$ for each graph $G_i$ vary
with $i$, we must first convert the
problem into an average reachability format.

\paragraph{Conversion to an average reachability problem} \label{sect:tau}
Suppose we have sampled as above
$\ell$ pairs of sample graphs and
external seeds: $\{ (G_1, A_1^{ext}), \ldots, (G_\ell , A_{\ell}^{ext} ) \}$
To compute  
$\hat{\sigma}(A) = \frac{1}{\ell} \sum_{i = 1}^\ell \sigma_i ( A \cup A_{i}^{ext})$
efficiently, we first convert this sum to a
generalized reachability problem:
we construct graph $H_i$ from
$G_i$ by removing all nodes (and incident edges)
reachable from $A_i^{ext}$: 
$H_i = G_i - \sigma_i ( A_i^{ext} )$. 
The average reachability of a set $A$
in the graphs $\{H_i \}$, which we 
term $\tau (A)$, is
\begin{equation}
  \tau (A) = \frac{1}{\ell} \sum_{i=1}^\ell \tau_i (A), \label{eq:tau}
\end{equation}
where $\tau_i (A)$ is the size of the set reachable
from $A$ in $H_i$. The estimators formulated by Cohen et al.
are suitable to estimate the value of $\tau$,
and the following two lemmas
show how we can compute $\hat{ \sigma }(A)$ from
$\tau (A)$.
\begin{lemma} \label{lem:reach}
  The size of the reachable set from $A$ in $H_i$
  can be computed from reachability $\sigma_i$ in $G_i$
  as follows:
 $\tau_i (A) = \sigma_i (A \cup A_i^{ext}) - \sigma_i (A_i^{ext}).$
\end{lemma}
\begin{proof}
  Suppose, in $G_i$, 
  $x$ is reachable from $A$, but not from 
  $A_i^{ext}$. This is true iff 
  there exists a path from $A$
  to $x$ in $G_i$ avoiding $\sigma_i (A_i^{ext})$,
  which is equivalent to the path existing in
  in $H_i$.
\end{proof}
The next lemma shows explicitly 
how to get $\hat{ \sigma }(A)$
from $\tau (A)$:
\begin{lemma} \label{lem:sig_tau}
  $\tau (A) = \hat{ \sigma }(A) - 
  \frac{1}{\ell} \sum_{i=1}^\ell \sigma_i 
  ( A_i^{ext} ).$
\end{lemma}
\begin{proof}
  This statement follows directly from Lemma
  \ref{lem:reach} and the definitions of
  $\hat{\tau}$, $\hat{\sigma}$.
\end{proof}
In Alg. \ref{alg:bicriteria}, for each $i$, $\sigma_i (A_i^{ext} )$
is computed in the construction of $H_i$; its size can be stored
as instructed on line 5, and used in the computation of
$O$ for line 14 in the greedy 
algorithm's stopping criterion.

\paragraph{Estimation of $\tau (A)$, method $C1$} \label{sect:est_tau}
In this subsection,
we utilize methods developed by Cohen et al.
\cite{Cohen2014a} to estimate efficiently
the average reachability problem $\tau (A)$
defined in (\ref{eq:tau}). For convenience,
we refer to this method of estimation in
the rest of the paper as method $C1$. 

Each pair
$(v, i) \in V \times \{1, \ldots, \ell \}$,
consisting of a node $v$
in graph $H_i$, is assigned
an independent, random rank value
$r_v^{(i)}$ uniformly chosen on the interval $[0,1]$.
The \emph{reachability sketch} for a set $A$ is
defined as follows: let $k$ be an integer,
and consider
$X_A = \text{bottom-$k$} \left( \{ r_v^{(i)} | (v, i) \in \tau_i (A) \} \right),$
where \emph{bottom-$k$}(S) means to take the 
$k$ smallest values of the set $S$, and
$R^{(i)}_A$ is the set of nodes reachable 
from $A$ in graph $H_i$.
The \emph{threshold rank} of a set $A \subset V$ 
is then defined to be $\gamma_A = \max X_A,$
if $|X_A| = k$, and $\gamma_A = 1$ if $|X_A| < k$.
The estimator for $\tau (A)$ is
then
$\hat{ \tau }(A) = (k - 1) / (\ell \gamma_A) $

If $k = (2 + c) (\epsilon)^{-2} \log n$,
 the probability that this estimator
has error greater than $\epsilon$ is at most $1 / n^c$
\cite{Cohen2014a}. The bounds on $k$ 
needed for the proof of Theorem
\ref{thm:guarantee} are determined by 
taking $\epsilon = \alpha T$.

\paragraph{Computation of $X_A$}
First, we compute $X_u$ for all $u \in V$ using
Algorithm 2 of
\cite{Cohen2014a} in time $O( k \ell m )$,
where $m$ is the maximum number of edges
in any $H_i$. This collection
$\{ X_u : u \in V \}$ is referred to as
\emph{oracles}.

Next, we discuss how to compute, for an
arbitrary node $u$, $X_{A \cup \{ u \} }$
given that $X_A$ has already been computed.
This computation will take time $O(k)$ and
is necessary for the
bicriteria algorithm: given that $X_u$ and 
$X_A$ are both sorted,
we compute $X_{A \cup \{ u \}}$ by merging
these two sets together until
the size of the new set reaches $k$ values.

\paragraph{Estimation of $\tau (A)$, method $C2$} \label{sect:est_C2}
Alternatively to estimator $C1$, we can estimate
$\tau (A)$ from the oracles $\{ X_u : u \in V \}$
in the following way. Let $\gamma_u$ be the
threshold rank as defined above for $X_u$. 
Then
$\hat{ \tau } (A) = \frac{1}{\ell}
\sum_{z \in \bigcup_{v \in A} X_v \backslash \{ \gamma_v \}} \frac{1}{ \max_{u \in A | z \in X_u \backslash \{ \gamma_u \} } \gamma_u }.$
For convenience, we refer to the estimator in the rest of
the paper as estimator $C2$; it was originally introduced in
\cite{Cohen2009}.

\section{Experimental evaluation} \label{sect:exp}
In this section, we demonstrate the scalability of
STAB as compared with the current state-of-the-art
IMM algorithm \cite{Tang2015} and with the
greedy algorithm for TAP in Goyal et al. \cite{Goyal2013}.
The methodology is described
in Section \ref{sect:exp:methods}, comparison to existing
IM algorithms is in Section \ref{sect:exp:comp}, and investigation
of the effect of external influence on the performance
of STAB is in Section \ref{sect:exp_ext}.
All experiments
were run on an 
Intel(R) Core(TM) i7-3770K CPU @ 4.0GHz
CPU with 32 GB RAM.
\subsection{Datasets and framework} \label{sect:exp:methods}
We evaluated the following algorithms
in our experiments.

\subsubsection{CELF} the greedy algorithm by naive sampling
of Kempe et al. \cite{Kempe2003}
can be modified to find a solution to TAP, as
shown in Goyal et al. \cite{Goyal2013}. The 
modified algorithm
performs Monte Carlo sampling at each step to select 
the node with the highest marginal gain into the seed set
until the threshold $T$ is satisfied.
The Cost-Effective Lazy Forward (CELF) approach
by Leskovec et al. \cite{Leskovec2007} improves
the running time of this algorithm by reducing the number
of evaluations required. 

\subsubsection{IMM}
The IMM algorithm \cite{Tang2015} is the current
state-of-the-art algorithm to solve the IM problem,
where the number of
    seeds $k$ is input. Since TAP asks to minimize the number
    of seeds, this algorithm cannot be applied directly.
    For the purpose of comparison to our methods,
    we adapt the algorithm by performing a binary search on
    $k$ in the interval $[1, T]$, 
    where $T$ is the threshold given in
    TAP. At each stage of the search, IMM utilizes the 
    value of $k$ in question until the minimum $k$ as estimated
    by IMM is found. Since binary search can identify the 
    minimum in at most $\log T$ iterations, we chose this 
    approach over starting at $k = 1$ and incrementing by one 
    until the minimum is found, which in the worst case would
    require $T$ iterations.

\subsubsection{STAB} the STAB algorithm (Alg. \ref{alg:bicriteria})
    using estimators $C1$ and $C2$, referred to as
    STAB-C1, STAB-C2 respectively. Since these 
    are greedy algorithms with an approximately 
    submodular function, we also use the CELF
    approach to reduce the number of evaluations
    performed by STAB-C2; for STAB-C1, we found
    this optimization unnecessary.

\emph{Network topologies:}
We generated networks
according to the
Erdos-Renyi (ER) and Barabasi-Albert (BA) models.
For ER random graphs, we used varying number
of nodes $n$; the independent probability that an edge
exists is set as $p = 2 / n$.  The BA model 
was used to generate scale-free synthesized
graphs; the exponent in the power law
degree distribution was set at
$-3$ for all BA graphs.

The
following topologies of real networks collected
by the Stanford Network Analysis Project \cite{snapnets} were utilized:
1) Facebook, a section of the Facebook network,
with 
$n = 4039$, $m =88234$;
2) Nethept, high energy physics collaboration network,
$n = 15,229$, $m = 62,752$;
3) Slashdot, social network with $n = 77,360$, $m = 828,161$;
4) Youtube, from the Youtube online social network,
$n = 1,134,890$, $m =2,987,624$; and
5) Wikitalk, the Wikipedia talk (communication)
network, $n = 2,394,385$, $m = 5,021,419$.


\begin{table}[h]
\centering

\begin{tabular}{| l | c | c | c | c |}
  \hline
  &  $\alpha = 0.1$ & $\alpha = 0.2$ & $\alpha = 0.3$ & $ip_{max}$\\
  \hline
  ER 1000 & 1.0 & 0.4 & 0.2 & 1\\
  \hline
  BA 15000 & 6.4 & 1.6 & 0.9 & 0.1 \\
  \hline
  Nethept  & 954 & 230 & 102 & 1\\
  \hline
  Slashdot & 25 & 6.4 & 3.3 & 0.01 \\
  \hline
  Youtube  & 385 & 91 & 41 & 0.01 \\
  \hline
  Wikitalk & 1704 & 274 & 122 & 0.01 \\
  \hline
\end{tabular}
\vspace{+0.2cm}
\caption{Oracle computation time (sec)} \label{table:oracles}
\shrinkfig
\end{table}
\emph{Models of internal and external influence:}
In all experiments, we used the independent
cascade model for internal influence propagation:
each edge $e$ in the graph is assigned 
a uniform probability
$ip_e \in [0, ip_{max}]$. For synthesized networks, we
usually set $ip_{max} = 1$; for real networks, we
observed that if $ip_{max} = 1$, then in most cases
a single node can activate a large fraction of the
network (often more than 33\%). However,
a large number of empirical studies have confirmed
that most activation events occur within a few hops
of a seed node 
\cite{Goel2015,Gonzalez-Bailon2013,Myers2014}; 
these works indicate that $ip_{max} = 1$ is not a realistic
parameter value for the IC model.  Therefore, we also ran experiments 
using lower values for
$ip_{max}$.


For external influence, we adopted in all experiments
the model that each node $u$ is activated externally
independently with uniformly chosen probability
$ep_u \in [0, ep_{max}]$. The setting of $ep_{max}$
is discussed in the context of each experiment.
Unless otherwise stated, we set $\delta = 0.01$,
which we found sufficient to return a solution
within the guarantee provided in Theorem \ref{thm:guarantee}
in most cases with estimator $C1$, and in all cases
with estimator $C2$.

\begin{figure*}
  \subfigure[ER, activation] {
    \includegraphics[width=0.22\textwidth]{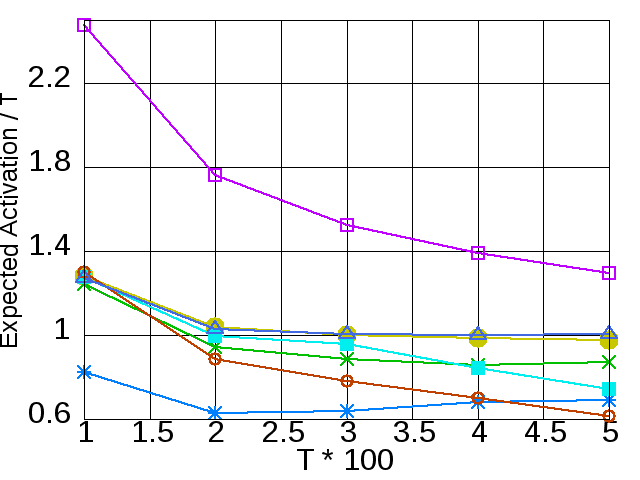}
    \label{fig:er_act}
  }
  \subfigure[Nethept, activation] {
    \includegraphics[width=0.22\textwidth]{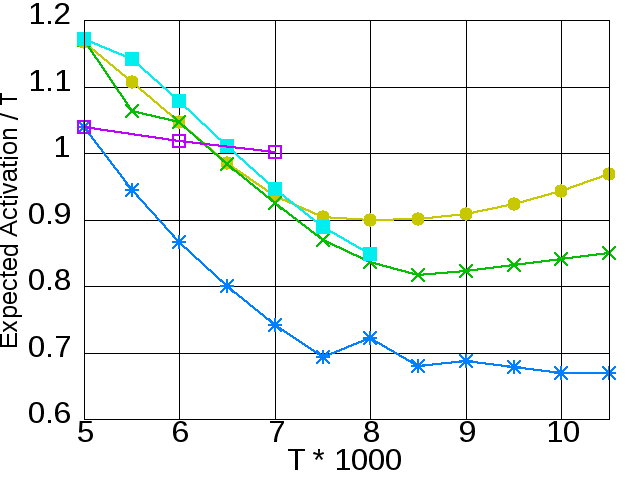}
    \label{fig:nh_act}
  }
  \subfigure[Youtube, activation] {
    \includegraphics[width=0.22\textwidth]{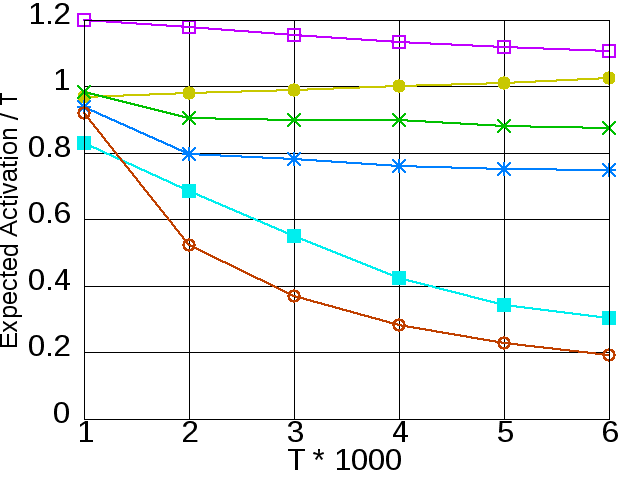}
    \label{fig:yt_act}
  }
  \subfigure[Wikitalk, activation] {
    \includegraphics[width=0.22\textwidth]{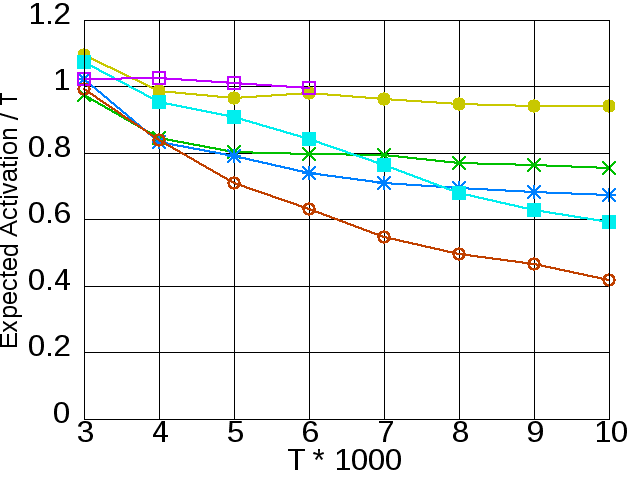}
    \label{fig:wiki_act}
  }
  
  \subfigure[ER, running time] {
    \includegraphics[width=0.22\textwidth]{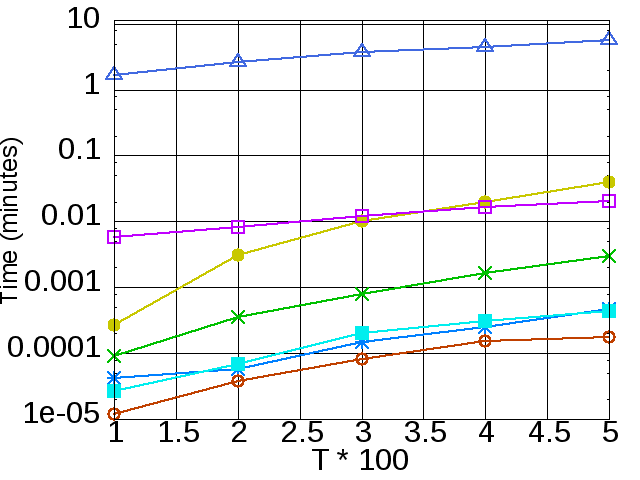}
    \label{fig:er_rt}
  }
  \subfigure[Nethept, running time] {
    \includegraphics[width=0.22\textwidth]{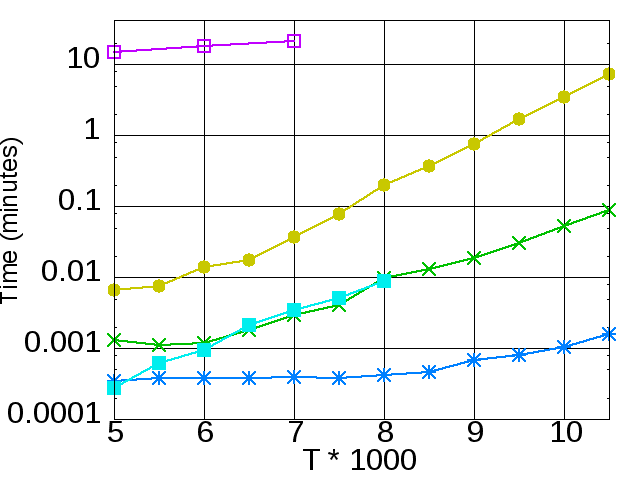}
    \label{fig:nh_rt}
  }
  \subfigure[Youtube, running time] {
    \includegraphics[width=0.22\textwidth]{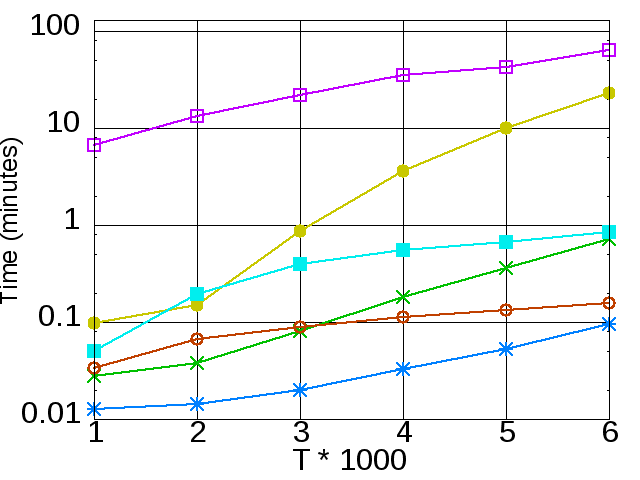}
    \label{fig:yt_rt}
  }
  \subfigure[Wikitalk, running time] {
    \includegraphics[width=0.22\textwidth]{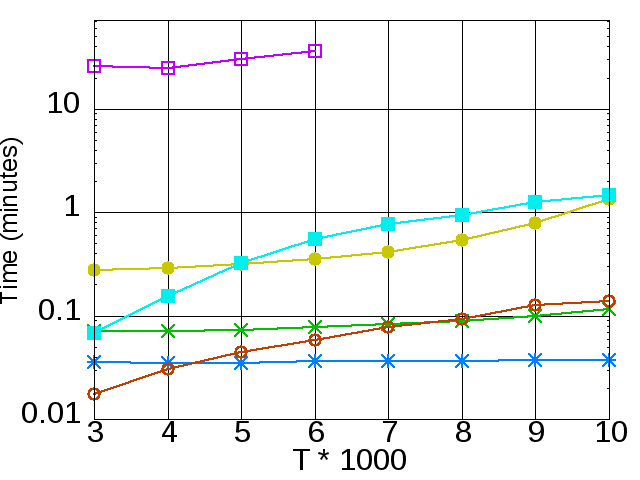}
    \label{fig:wiki_rt}
  }

  \subfigure[ER, number of seeds] {
    \includegraphics[width=0.22\textwidth]{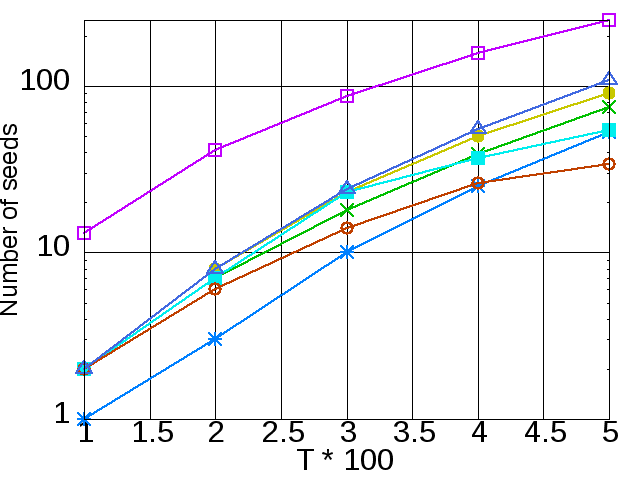}
    \label{fig:er_ss}
  }
  \subfigure[Nethept, number of seeds] {
    \includegraphics[width=0.22\textwidth]{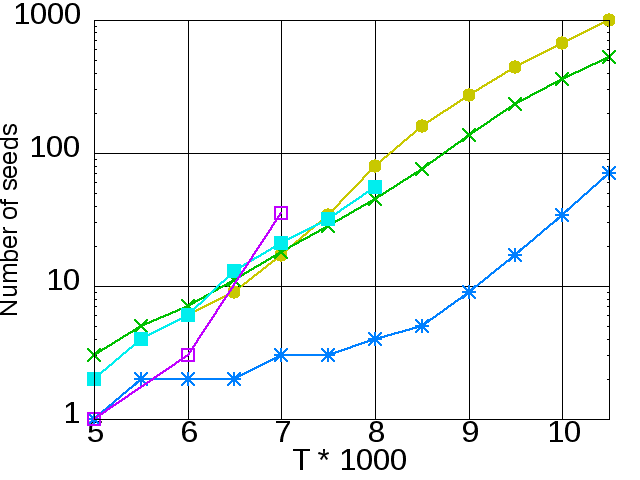}
    \label{fig:nh_ss}
  }
  \subfigure[Youtube, number of seeds] {
    \includegraphics[width=0.22\textwidth]{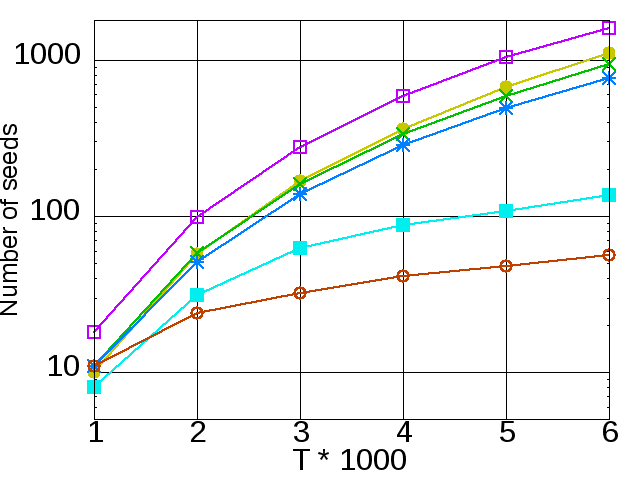}
    \label{fig:yt_ss}
  }
  \subfigure[ Legend ] {
    \includegraphics[width=0.22\textwidth]{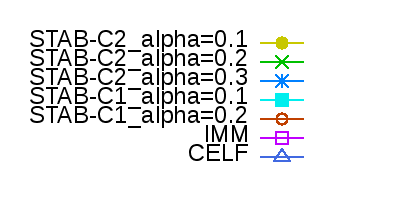}

  }
  \caption{ First row: expected activation normalized by threshold $T$ versus $T$. Second row: running time in minutes versus $T$. Third row: size of seed set versus $T$.} \label{fig:comparison}   \shrinkfig
\end{figure*}

\subsection{Performance comparison and demonstration of scalability} \label{sect:exp:comp}
In this section, we compare the performance of 
STAB-C1 and STAB-C2 to  CELF and IMM
as described above;
we experimented 
on the datasets listed in Table \ref{table:oracles},
where the total CPU time required 
to compute the oracles is shown. The
oracle computation is parallelizable, and in our
experiments we used 7 threads of computation.
The oracles for each value of $\alpha$ were computed
once and stored, and thereafter when running STAB the
oracles were
simply read from a file. The running
times we report for the various versions of STAB do
not include the oracle computation time unless otherwise
specified. Also in Table \ref{table:oracles} are the
values of $ip_{max}$ used for each dataset in this
set of experiments. Since IMM and CELF do not consider
external influence, the experiments in this section
had no external influence; that is, $ep_{max} = 0$
for all datasets. To evaluate the seed set returned
by the algorithms we used the average activation
from 10000 independent Monte Carlo samples.

We show typical results in Fig. \ref{fig:comparison} 
on the following
four datasets: ER 1000 ($n = 1000$),
Nethept, Youtube, and Wikitalk.
The first row of Fig. \ref{fig:comparison} 
shows the expected activation,
normalized by the threshold value $T$,
of the seed set returned by each algorithm plotted against
$T$. Thus, a value of $1$ indicates the algorithm
successfully achieved the threshold of activation. 
The second row of the figure shows the running time 
in minutes of
each algorithm against $T$, and the third row shows the
size of the seed set returned by each algorithm.

For the ER 1000 network results are shown in the first
column of Fig. \ref{fig:comparison}; this dataset was
the only one on which CELF finished under the time limit
of 60 minutes.
We see that IMM is consistently returning a larger seed set
than STAB and CELF and overshooting the threshold value
in activation by as much as a factor of 2.5. Thus, it
has poor performance of minimizing the size $k$
of the seed set for TAP. This behavior 
appears to be a result of IMM
underestimating the influence of its seed set internally.
As expected, CELF performs very well in terms of the size
of the seed set and meeting the threshold $T$, but is
running on a timescale larger than the other algorithms
by factor of at least 100 and as large as
$10^4$. Notice that STAB-C2 with $\alpha = 0.1$ 
has virtually identical activation and size of
seed set to CELF, while running at a much faster
speed; as expected, the quality of solution of STAB
deteriorates as $\alpha$ is raised, but the running time
decreases drastically. In addition, notice that none of
the versions of STAB seed too many nodes and overshoot
the threshold as IMM does; instead, STAB errs on the side
of seeding too few nodes and only partially
achieving the threshold $T$ of activation.
Finally, it is evident that STAB-C1 runs faster than
STAB-C2 and has a similar amount of activation when
the seed set required is relatively small. However,
the larger the seed set required, the farther
is STAB-C1 from achieving the threshold $T$ while
STAB-C2 does not suffer from this drawback.

The results from Nethept are shown
in the second column of Fig. \ref{fig:comparison}; 
with $n = 15,229$, CELF was unable to run
within the time limit, and 
IMM was able to complete its binary search only
for $T \le 7000$; for the higher threshold values,
IMM exceeded the 32 GB memory usage limit imposed 
in our setup. On this network, STAB-C2 again 
exhibits the best performance. Despite
an initial decrease in activation relative to $T$
shown in Fig. \ref{fig:nh_act}, the trend reverses
around $T = 8000$ and seed set size $100$ and the
algorithm gets closer to achieving the threshold.
This behavior is explained by lower coefficient
of variation of
estimator $C2$ as analyzed in \cite{Cohen2009};
the CV can be lower than estimator $C1$ by up to a factor
of $\sqrt{ |A| }$, where $A$ is the seed set.
In stark contrast, STAB-C1 was unable to proceed past $T = 8000$
even for $\alpha = 0.1$ because the estimator $C1$ 
appears to lose accuracy as the number of seeds increases.
By this mechanism, STAB-C1 had achieved its maximum
estimation of influence of any set and thereby could not
increase it before reaching an estimated activation
of $T$, for $T > 8000$.

Next, the Youtube network is
shown in the third column
of Fig. \ref{fig:comparison}.
As in the ER network, IMM
is underestimating the influence of its seed set
and thereby picking too many seeds and
overshooting the threshold, as shown in
Fig. \ref{fig:yt_act}; IMM picks nearly
twice as many seed nodes as STAB-C1, $\alpha = 0.1$,
as shown in Fig. \ref{fig:yt_ss}.
In Fig. \ref{fig:yt_rt},
the scalability of STAB is demonstrated as the most
precise version, STAB-C2 with $\alpha = 0.1$, 
runs faster than IMM by as much as a factor of 50.
On our largest dataset, the Wikitalk network, IMM
again exceeded 32 GB memory after $T = 6000$;
notice that the running time of STAB in all cases
is under 2 minutes and STAB-C2 with $\alpha = 0.3$
maintained activation greater than $0.6 T$ while
running in less than 5 seconds, as shown in 
Fig. \ref{fig:wiki_rt}. With inclusion of the parallelizable
and reusable oracle computation time of 122 seconds from
Table \ref{table:oracles}, the total time taken
by STAB-C2 is less than 3 minutes. 
The total running time for $\alpha = 0.1$ of STAB-C2
including the oracles
at $T = 10000$ is less than 30 minutes. 

\emph{Choice of $\alpha$:} The above discussion
demonstrates that $\alpha = 0.1$ provides the
close activation to the threshold $T$ while
maintaining high scalability. If faster running time
is desired, $\alpha$ may be increased, which results
in a loss of accuracy shown clearly in Fig. \ref{fig:wiki_act},
for $\alpha \in \{ 0.1, 0.2, 0.3 \} $. On the other hand, 
if activation closer to $T$ is required, smaller values of 
$\alpha$ may be used at higher computational cost.

\shrinkfig
\subsection{External influence} \label{sect:exp_ext}
In this section, we analyze the performance of STAB
when external influence is present in the network;
that is, when $ep_{max} > 0$. For this section, we
considered a BA network with 100,000 nodes, with
a threshold of $T = 1000$, and the Facebook network
with $T = 1500$.
Results on other topologies were qualitatively similar.
For all experiments in this section, 
we set $\alpha = 0.1$.

In \ref{fig:BA_ep}, we plot the expected activation
(Act) of the seed set returned by STAB-C2 normalized
by the threshold $T$, as $ep_{max}$ varies from
0 to 0.006.  As in the previous section,
the expected activation of the returned 
seed set $A$ is estimated
by independent Monte Carlo sampling. We also plot the
expected fraction of $T$ activated by the external
influence, along with the size of the seed set returned
by STAB-C2, normalized by its maximum value. 
As the effect of external influence in the network
increases, the algorithm requires fewer
seed nodes to ensure the expected activation is within
the specified error tolerance to threshold $T$.

In \ref{fig:FB_ep}, we show an analogous plot for the
Facebook network; interestingly, the size of the seed set chosen
by STAB-C2 nearly doubles as $ep_{max}$ increase from
0 to 0.01, before beginning to decrease to 0. This
increase differs from what we expected; it is counterintuitive
that the algorithm would require more seeds to reach
the threshold as external influence increases. One 
possible explanation for this effect
is that the external influence both decreases and
distributes the marginal gain more evenly among 
seed nodes, so that the
greedy algorithm has a more difficult time identifying
the best seed nodes, especially in the presence of
the error of estimation. \shrinkfig


\begin{figure}
  \subfigure[BA 100,000] {
    \includegraphics[width=0.22\textwidth]{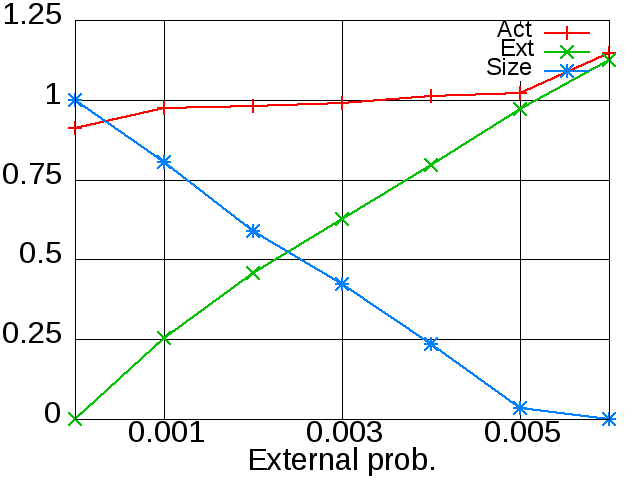}
    \label{fig:BA_ep}
  }
  \subfigure[Facebook] {
    \includegraphics[width=0.22\textwidth]{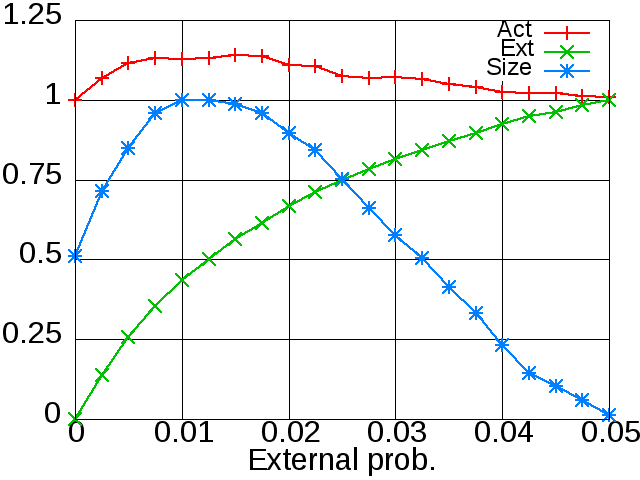}
    \label{fig:FB_ep}
  }

 
  \caption{The effect of external influence on activation and size of 
    seed set of STAB-C2.} \shrinkfig
\end{figure}

\section{Related work} \label{sect:rw}
Kempe et al. introduced the triggering model 
in a seminal work on the IM problem \cite{Kempe2003}, where
they exhibited a Monte Carlo greedy sampling algorithm
that achieves $1 - 1/e$ performance ratio for IM; this 
algorithm, although it runs in polynomial time, is very
inefficient and cannot scale well. Since the maximum coverage
problem prohibits performance guarantee better than
$1 - 1/e + o(1)$ under standard assumptions, 
the ratio for IM likely cannot
be improved, but much work has improved the scalability
of the algorithm while retaining the guarantee.
Leskovec et al. \cite{Leskovec2007} introduced the CELF method
to exploit submodularity and improve the running time.
Reverse Influence Sampling (RIS) was introduced in
\cite{Borgs2014} to
further improve the greedy performance; algorithms using
RIS include \cite{Borgs2014,Tang2014,Tang2015,Nguyen2016a,Nguyen2016};
the current state-of-the-art are the SSA \cite{Nguyen2016a} and the IMM algorithm \cite{Tang2015}
to which we compare in this paper.
Cohen et al. \cite{Cohen2014a} introduced a 
methodology for a highly scalable IM algorithm SKIM;
in this work, we extend this methodology to solve the TAP
problem with the triggering model and external influence.

As compared with IM, 
much less effort has been devoted to scalable solutions
to TAP while maintaining performance guarantees;
Goyal et al. \cite{Goyal2013} studied the TAP problem 
with monotonic and submodular models of influence propagation;
their bicriteria guarantees differ from ours, and provide
no method of efficient sampling required for scalability.
Chen et al. \cite{Chen2015}  considered external influence in a viral marketing context.
However, their model of external influence is much less
general than ours. Furthermore, they restrict external
influence to only pass through seed consumers and
have no discussion of sampling, scalability,
or the TAP problem. Nguyen et al. \cite{Nguyen2010,Nguyen2011}
studied methods to restrain propagation in social networks.
\vspace{-0.3cm}
\section{Conclusion}
We establish equivalency between the triggering
model and generalized reachability, which allows
incorporation of external influence into our efficient
sampling techniques. We gave precise trade-off between
accuracy and running time with a bound on the number of 
samples required to solve TAP. Our algorithm is highly
scalable and outperforms adaptations to TAP of the current
state-of-the-art algorithm for the IM problem.

\vspace{-0.1in}
\section*{Acknowledgment}
\vspace{-0.05in}
This work is supported in part by the NSF grant \#CCF-1422116.
\renewcommand{\baselinestretch}{.87}
\bibliographystyle{unsrt}
\bibliography{refs,refs2}

\end{document}